\documentclass[times, 10pt,latex8]{article}

\usepackage{amssymb}
\usepackage{amsfonts}
\usepackage{amsmath}

\usepackage{latexsym}
\usepackage{epsfig, epstopdf, wrapfig}
\usepackage{graphicx, xcolor}
\usepackage{caption}
\usepackage{enumerate}
\usepackage{mathrsfs}
\usepackage{amsthm}
\usepackage{hyperref}
\theoremstyle{plain}

\newtheorem{propn}{Proposition}{}{}
\newtheorem{thm}{Theorem}{}{}
\newtheorem{lem}{Lemma}{}{}
{}{}
\newtheorem*{cor*}{Corollary}{}{}
{}{}
\newtheorem*{Def*}{Definition}{}{}
{}{}
{}{}

\newcommand{\commentout}[1]{}

\newcommand{\Real}{\mathbb{ R}}
\newcommand{\Complex}{\mathbb{ C}}
\newcommand{\modulus}[1]{|\!#1\!|}
\newcommand{\norm}[1]{|\!| #1 |\!|}

\newcommand{\ket}[1]{|#1\rangle}
\newcommand{\bra}[1]{\langle #1|}
\newcommand{\inp}[2]{\langle #1|#2\rangle}
\newcommand{\inpr}[3]{\langle #1|#2|#3\rangle}

\renewcommand{\vec}[1]{{\bf #1}}

\newcommand{\conj}[1]{\overline{#1}}
\newcommand{\hconj}[1]{{#1}^{\dagger}}

\newcommand{\tr}{\text{\tt Tr}}
\newcommand {\pj}[1]{\ket{#1}\!\bra{#1}}

\newcommand{\acti}{\!\cdot\!}

\newcommand{\unitary}[1]{{\mathcal{U}(#1)}}
\newcommand{\unp}{\mathcal{U}}
\def\dmn#1#2{{#1^{(#2)}}} 
\def\hilbert#1{\mathcal{#1}}

\newcommand{\be}{\begin{enumerate}}
\newcommand{\ee}{\end{enumerate}}
\newcommand{\bi}{\begin{itemize}}
\newcommand{\ei}{\end{itemize}}
\newcommand{\beq}{\begin{equation}}
\newcommand{\eeq}{\end{equation}}
\newcommand{\beqx}{\begin{displaymath}}
\newcommand{\eeqx}{\end{displaymath}}
\newcommand{\beqa}{\begin{eqnarray}}
\newcommand{\eeqa}{\end{eqnarray}}
\newcommand{\beqax}{\begin{eqnarray*}}
\newcommand{\eeqax}{\end{eqnarray*}}

\title{Entangling capacity of operators}

\author{
Manas ~K. ~Patra\\ School of Computing, Data and Mathematical Sciences\\Western Sydney University\\ Victoria Road, Rydalmere, NSW 2116, Australia }
\date{}
\begin{document}
\maketitle

\begin{abstract}
Given a unitary operator $U$ acting on a composite quantum system what is the entangling capacity of $U$? This question is investigated using a geometric approach.  The entangling capacity,  defined via  metrics on the unitary groups, leads to a {\em minimax} problem. The dual, a  {\em maximin} problem,  is investigated in parallel and yields some familiar entanglement measures. A class of entangling operators, called generalized control operators is defined. The entangling capacities and other properties  for this class of operators is studied.   
\end{abstract}
\section{Introduction}
An earlier work by the author \cite{Patra1} proposed families of metrics and {\em pseudometrics} on the group of unitary operators in finite-dimensional Hilbert spaces. The current work investigates the application of these geometric objects to quantum entanglement. It is unnecessary to expound on the fascination or usefulness  of entanglement. For many it is the most mysterious phenomenon in quantum mechanics.  I  refer to some of the 
 excellent reviews on the topic of entanglement \cite{Horodecki09, Gunhe}. An important problem of entanglement is quantifying it. Simply put, the question is: given two (mixed) states $\rho$ and $\sigma$ on a composite or product quantum system how may we reasonably assert one of the states, say $\rho$ is more entangled than $\sigma$. The most common approach is to use an entanglement measure. The reviews mentioned above have extensive discussion on entanglement measures. One particular class  of entanglement measures is associated with the geometric notion of distance. All these entanglement measures (geometric or otherwise) refer to quantum states, both pure and mixed, in a product Hilbert space. However, we must keep in mind that any entangled state is an outcome of some interaction between the two initially independent/uncoupled subsystems characterised by  coupling the two subsystems. So, if we have to create entangled quantum states in the laboratory we must implement some interaction Hamiltonian between two subsystems and the corresponding time evolution operator gives us a one-parameter group of unitary operators (assumed time-independent Hamiltonians for simplicity). All this, of course, is textbook knowledge. But it raises the following question: {\em given a unitary operator $U$ on the product Hilbert space $H_1\otimes H_2$, how do we characterise the entangling capacity of $U$?}
The question is tackled from a geometric viewpoint in this paper. Informally, the {\em entangling  capacity}  measures how well unitary operators can be approximated by local (product) unitary operators. The `quality of approximation' is measured using distance metrics. The farther the operator to the manifold of product unitary operators higher is its capacity. 

In this paper I mostly restrict to the ideal noiseless case implying unitary evolution. In general, quantum operations, including measurements, are completely positive (CP) maps \cite{Nielsen}. The geometric approach of this paper can be extended to CP maps. A detailed summary of the paper follows. 

The main results of \cite{Patra1} are stated in Section 2. These are needed in the subsequent analysis. However, Subsection \ref{sec:gen_def_metric} is new. In definition of the distance between 2 operators on quantum states (see \eqref{eq:basic_metric1} the maximisation is done over all states. New metrics may be defined by restricting the states over which maximization is done to appropriate subsets.  This is done in Subsection \ref{sec:gen_def_metric}. These general metrics are used in the definition of entangling capacity. 

In Section 3, I define the {\em entangling capacity} $C(U)$ of a unitary operator $U$ acting on the tensor product $\mathcal{H}\otimes \mathcal{K}$ of 2 Hilbert spaces and discuss consequences. $C(U)$ is defined as the distance of $U$ to the group of local unitary operators. This is an ``intrinsic'' definition as $C(U)$ may be considered as a measure of approximating $U$ by local operations. The definition of $C(U)$ leads to a {\em minimax} problem \cite{Rockafellar}. Any minimax problem can be dualised in an obvious way---simply change the order of ``max''  and ``min'' operations. The ``dual capacity'' of $C(U)$ is denoted by $C_E(U)$, where $E$ is a geometric entanglement measure on states. Then $C_E(U)$ is  maximum  entanglement (as measured by $E$) possible among all states generated by $U$ from local states.  Moreover, the  minimax inequality \cite{Rockafellar} implies $C_E(U) \leq C(U)$. The question of equality of the 2 capacities is  discussed later. The definition of the most general class of capacities via the geometric approach may be obtained by varying the two sets over which the optimizations are carried out (see \eqref{eq:entcapG}).  

Section 4 deals with a special class of entangling unitary operators on $\mathcal{H}\otimes \mathcal{K}$, called \emph{generalized control} operators.  All control gates of quantum computing belong to this class. I discuss some properties of these operators and carry out some computations on their capacities . 

In the final section,  I discuss question of equality of the entangling capacity $C(U)$ and its dual $C_E(U)$ called {\em minimax} equality. Actually, this is the main problem  of a variety of minimax theorems\cite{Du}. This is a difficult problem here because of inadequate convexity and other properties of the quantities of interest. Perhaps some extension of the scope of the definitions to include  more general quantum operations could help with the question of minimax equality. These issues are discussed here. Finally, I suggest some potential applications in the study of {\em open} quantum systems using entangling capacities of operators. 

The question of entangling capacity of quantum operations has attracted the attention of researchers from the early days of quantum computing \cite{Zanardi, Bennett03, Leifer,  Chen1, Bauml}. These papers address different issues of entangling power, generally motivated by operational issues. They define entangling capacity (sometimes called entangling power) of $U$ via the states it generates from ``local'' states. This is similar to the dual capacity $C_E(U)$ in this paper. 
\section{Metrics and pseudometrics on operators} 
The main definitions and results of \cite{Patra1} are summarised in this section. First, let us fix some notations used in the rest of the paper.  Let $\mathcal{H}$ be a Hilbert space of dimension $m$ and $B(\mathcal{H})$ the space of operators on it. The following notations will be used consistently. 

\bi
\item
$L(\mathcal{H}) \subset B(\mathcal{H})$ is the set of hermitian operators on $\mathcal{H}$.
\item
$\unitary{\mathcal{H}}$ will denote the group of unitary operators on $\mathcal{H}$.
\item
$S(\mathcal{H}) \subset L(\mathcal{H})$ is the convex set of positive semidefinite operators on $\mathcal{H}$ with trace 1. This is the set of (mixed) quantum states. 
\item
For $A\in L(\mathcal{H})$ and $\rho\in S(\mathcal{H})$ let $A\cdot \rho = A\rho \hconj{A} $ denote the action of $A$ on a state $\rho$. 
\item
$\mathrm{CP}(\mathcal{H})$ is the set of completely positive (CP) maps or superoperators on states. These are linear maps on $L(\mathcal{H})$. 
\item
If $\mathcal{H}$ and $\mathcal{K}$ are 2  Hilbert spaces of dimensions $m$ and $n$ respectively then $\Pi_p(\mathcal{H}\otimes \mathcal{K}) \subset S(\mathcal{H}\otimes \mathcal{K})$ will denote the set of pure product states and $\Pi_u(\mathcal{H}\otimes \mathcal{K}) \subset \mathcal{U}(\mathcal{H}\otimes \mathcal{K})$ the subgroup of product unitary operators. Members of  $\Pi_u$ will also be called local (unitary) operators. In case $m=n$, $\Pi_u$ will also denote its (semidirect) extension by the swap operator $S_w(\alpha\otimes \beta) =  \beta\otimes \alpha$.
\item
The convex hull of $\Pi_p$ denoted by $S_{sep}( \mathcal{H}\otimes \mathcal{K})$ is called the set of {\em separable} states. 
\item
In all cases, if the domain is clear from the context I drop it. Thus, write $S$ for $S(\mathcal{H})$, $\Pi_p$ for $\Pi_p(\mathcal{H}\otimes \mathcal{K})$, $\Pi_u$ in place of $\Pi_u(\mathcal{H}\otimes \mathcal{K})$ etc. 
\ei 
\subsection{The metrics and some properties}
Let $\Delta$ be a metric on the state space $S(\mathcal{H})$ and $\rho$ a fixed state. We assume that $\Delta$ is unitary invariant: $\Delta(U\cdot \mu , U\cdot \nu) = \Delta(\mu, \nu)$ for all $U\in \unp$ and $\mu, \nu \in S(\mathcal{H})$. Let $A, B: S(\mathcal{H})\rightarrow S(\hilbert{H})$ be CP maps. Define
\beq \label{eq:basic_pseudo}
\Delta_\rho(A, B) = \Delta(A\cdot\rho, B\cdot\rho) 
\eeq
\begin{propn}
$\Delta_\rho$ is a pseudometric on $\mathrm{CP}(\mathcal{H})$. It is unitarily (left) invariant in the sense that $\Delta_\rho(UE, UF) = \Delta_\rho(E, F)$. The function 
\beq \label{eq:basic_metric1}
\Delta(E, F) = \sup_\rho \Delta_\rho(E, F) 
\eeq
is a unitarily invariant (both left and right) metric on $\mathrm{CP}(\mathcal{H})$.  
\end{propn}
Recall that a pseudometric satisfies all the properties of a metric (symmetry, triangle inequality) but is only positive {\em semidefinite}. All the statements are easily proved (\cite{Patra1}). 

\commentout{
For unitary operators, $X,Y$,  $\Delta_\rho(X,Y) = 0$ if and only if $\rho$ satisfies $\hconj{X}Y\acti\rho = \rho$. Hence, if the action of the unitary group is linear (on the space of operators) which we also assume, then $\rho$ is an eigenstate of $\hconj{X}Y$ with eigenvalue 1 or alternatively, it is a {\em fixed point}. 
If $X\in \mathcal{U}(\mathcal{H})$ is a unitary operators, its action on $S(\mathcal{H})$  is defined by $\rho\rightarrow U\cdot\rho = X\rho \hconj{X}$.  Then $ \mathcal{U}(\mathcal{H}) \subset CP(\mathcal{H})$. 
}
So any unitary invariant metric on  $S(\mathcal{H})$ will induce a corresponding pseudometric on $CP(\mathcal{H})$ with the above properties. In particular, consider the  {\em trace} norm (also called $L_1$ norm): $\norm{A}_1 = \tr (|A|)$ (the trace of $|A|$) where $|A| = \sqrt{\hconj{A}A}$.\footnote{$\norm{A}_1$ is the sum of singular values of $A$ \cite{Bhatia}.}
\def\dmn#1#2{{#1^{(#2)}}}
The trace-norm induces a unitary invariant metric 
\[\dmn{d}{1}(\rho, \sigma) = \frac{\norm{\rho - \sigma}_1}{2}\]
which has close analogy with classical probability distance \cite{Nielsen}. We  use $d_\rho$ to denote the pseudometric induced by the metric $\dmn{d}{1}$ in the state $\rho$: 
\beq\label{eq:tr_dist}
d_\rho(A, B) = \dmn{d}{1}(A\cdot \rho, B\cdot \rho)
\eeq
for CP maps $A,B$. In particular, for  unitary operators $U,V$ unitary invariance of $\dmn{d}{1}$ implies $d_\rho(U,V) = d_\rho(I, U^\dagger V)$. If $\rho = \pj{\psi} $ is a pure state then write $d_\psi$ instead of $d_{\pj{\psi}}$. Using known properties of the trace norm \cite{Nielsen}
\beq \label{def:psi-dist}
d_{\psi} (U,V) = (1 - |\inpr{\psi}{U^{\dagger}V}{\psi}|^2 )^{1/2}
\eeq

In general, we can replace $\sup$ (supremum) by $\max$ (maximum) in \eqref{eq:basic_metric1} since $S(\mathcal{H})$ is compact and the function  $\rho \rightarrow \Delta_{\rho}(A ,B)$ is continuous for fixed $A, B$. Suppose now the metric $\Delta$ is convex function on $\mathcal{D}$:
\[ 
\begin{split}
 \Delta(\lambda \rho_1 + (1 - \lambda) \rho_2, \lambda \sigma_1 + (1 - \lambda) \sigma_2)  \leq & \;
  \lambda \Delta(\rho_1 , \sigma_1)  + ( 1 - \lambda) \Delta ( \rho_2, \sigma_2), \\
& 0 \leq \lambda \leq 1. \\
\end{split}
\]
Then the maximum occurs at an extreme point of $S(\mathcal{H})$, that is, a pure state. This follows easily from the fact that any state $\rho$ is a convex combination of pure states and convexity of $\Delta$. In particular, if the metric is induced by a norm on $B(\mathcal{H})$ then it is convex. Hence if $\Delta$ is the trace distance ($L_1$ metric) or the Frobenius distance ($L_2$ metric) the maximum of the corresponding induced metric $d_\rho(E, F)$ occurs at a pure state.

Henceforth,  I assume  that $\Delta$ is induced by the trace or $L_1$ norm. This fixing will remain true in rest of the paper unless stated otherwise. Observe that, $d$ is {\em not} a metric on the group $\mathcal{U}(\mathcal{H})$ because it is not necessarily positive for distinct elements. For example, $d(U, cU) = 0, |c|= 1$. However, we continue to call it metric on $\mathcal{U}(\mathcal{H})$ with the understanding that we identify $U$, $cU$. Actually it  is a metric on the projective group $\text{P}\mathcal{U}_n = \mathcal{U}_n/Z(\mathcal{U}_n)$ ($n=\dim(\mathcal{H})$), the quotient group obtained by factoring out the centre. I state some  useful properties the metric $d$ from \cite{Patra1}. 
\begin{thm}\label{thm:properties}
For any pair of unitary matrices $U,V$ of order $n$, the function $d(U,V)$ is a pseudometric satisfying the following. 
\be[{\em 1.}]
\item {{\em Projective invariance.}} 
 For any complex number $c$ of modulus 1,
\[ d(U,cV) = d(cU,V) = d(U,V) \]
\item
$d(U, V) = 0$ if and only if $U = cV$ for some complex number $c$ with modulus 1. 
\item {{\rm Unitary invariance. }}
$d$ is invariant under left and right multiplications in the group
  $\mathcal{U}_n$.
\item
$d(U, V) = 1$ iff there is a unit vector $\alpha$ such that $U\alpha$
  and $V\alpha$ are orthogonal. 
\item
For unitary operators $U, V, W, X$
\beq\label{eq:monotone}
d(UV, WX) \leq d(U, W) + d(V, X)
\eeq
\ee
\end{thm}

\begin{thm} \label{thm:eigenChar}
Let $z_i = e^{i\theta_i}, \; i= 1,\ldots n$ be the eigenvalues (possibly
with repititions) of a unitary operator $W$. Let $\theta_i$'s be ordered
such that $0\leq \theta_1\leq \theta_2\leq \cdots \leq \theta_n < 2\pi$. Let $C$ be smaller arc containing all the points $z_i$. There are two alternatives. Let $\alpha = l(C)$, the length of $C$. It is the angle subtended at the center. 
\be
\item
$C$ lies in the interior of a semicircle and hence $\alpha < \pi$. Then
\beq \label{eq:basic_formula}
d(I, W) = \sin{(\alpha/2)}
\eeq
\item
$C$ contains a semicircle: $\alpha \geq \pi$. Then $d(I, W) = 1$. 
\ee
\end{thm}
First note that, if we set $W=\hconj{U}V$ then $d(I,W)=d(U,V)$ in \eqref{eq:basic_formula} above. The theorem  has some simple but useful consequences for tensor product of operators. Making use of the fact that the eigenvalues of the operator $A\otimes B$ are product of eigenvalues of $A$ and $B$ and the theorem we infer the following. 
\begin{propn}
Let $U,V \in \mathcal{U}(\mathcal{H})$ and $W, X \in  \mathcal{U}(\mathcal{K})$. Let $d(U, V) = \delta_1$ and $ d(W, X) = \delta_2$. Let $\delta = d( U\otimes W, V\otimes X)$, where the metric $d$ is now defined on $\mathcal{U}(\mathcal{H}\otimes\mathcal{K})$. Then 
\[
\delta =  
\begin{cases}
  \delta_1\sqrt{1 - \delta_2^2} + \delta_2 \sqrt{ 1 - \delta_1^2} \text{ \sf{if} } \delta_1^2 + \delta_2^2 < 1 \\
 1 \text{ otherwise.}
\end{cases}
\]
\end{propn}
The metrics on the operators defined here are induced by metrics on the space in which they act. In the case where the later are induced by norms we have essentially distances defined by the {\em operator norms}. This is the case for the $L_1$ or trace norm induced metric defined here. For some calculations, it may be better to choose other distances, for example Bures metric \cite{Bengtsson}. 
\subsection{A general class of metrics}\label{sec:gen_def_metric}\label{subsec:gen1}
In the definition of metric on $\mathcal{U}(\mathcal{H})$ the supremum is taken over {\em all} states in $S(\mathcal{H})$  (see \eqref{eq:basic_metric1}). One may, however, choose to restrict the set over which maximization is done to a proper subset of states. This may be considered as a restriction on the input resources. Let 
\beq\label{eq:gen_metric1}
d_K (U, V)  = \max_{\rho \in K} d_\rho(U, V), \; K \subset S(\mathcal{H}) 
\eeq
Here it is assumed that $K$ is a closed subset of $S(\mathcal{H})$. 
First observe that, $d_\rho(U, V)$ is a special case corresponding to $K = \{\rho\}$. The conditions of symmetry and triangle inequality for a metric are satisfied by $d_K$. Positivity may be ensured if, for example, $K$ spans the real vector space of hermitian operators over $\mathcal{H}$. An example of such a subset $K$ is the set of pure states (1-dim projection operators). Then an arbitrary state  $\sigma = \sum_i p_i \pj{\alpha_i}$ and 
\[
\begin{split}
2d_\sigma(U, V) & =\norm{U\cdot \sigma - V\cdot \sigma}_1 \leq \sum_i p_i \norm{U\pj{\alpha_i}\hconj{U} - V\pj{\alpha_i}\hconj{V}}_1 \\
& =2 \sum p_i d_{\alpha_i} (U, V) = 0\text{ if } \sup_{\alpha} d_\alpha(U,V) = 0\\
\end{split}
\]
Note that $K \subset L \subset S(\mathcal{H})$ implies $d_K (U, V) \leq d_L(U, V)$ for all $U, V \in  \mathcal{U}(\mathcal{H})$. This generalization of the metric is necessary for the definition of entangling capacities. 

\section{Entangling capacity of operators}
The metrics on operators discussed in the previous section may be used to define quantitative measures for the {\em entangling  capacity} of operators. Quantum entanglement is associated with the states of composite systems. These states give rise to correlations which cannot be generated by ``classical'' operations (LOCC). Therefore, we have to distinguish the states generated by certain ``quantum operations'' from those that can be created by classical operations. Consider an ideal  noiseless scenario. To create an entangled state from 2 subsystems we need an interaction Hamiltonian. The joint system then evolves under a unitary operator on the product space.  So it seems reasonable to ask how much entanglement can an interaction create? In other words, what is the {\em entangling capacity} of the interaction operator? Instead of dealing with an interaction Hamiltonian $H$, we could work with the evolution operator $U(\Delta t) = e^{iH\Delta t/\hslash}$  (with time interval $\Delta t$)\footnote{Time independent Hamiltonian assumed here.}. Starting from an initial product state, $\ket{\phi}\otimes \ket{\psi}$ (for simplicity, taken to be  pure) how much entanglement is there on $U(\ket{\phi}\otimes \ket{\psi})$ where $U$ is a unitary operator acting on the composite system? Here we have assumed ideal conditions: the exact Hamiltonian is known, there is no noise etc. In general the quantum operation is represented by a CP map. The definition of entangling capacity can be extended to CP maps. 

\subsection{Entangling capacity of unitary operators}
Let $\mathcal{H}$ and $\mathcal{K}$ be Hilbert spaces of dimensions $m$ and $n$ respectively and let $d_\pi$ be a metric on $\mathcal{U}(\mathcal{H}\otimes \mathcal{K})$ induced by a metric $d$ in  $B(\mathcal{H}\otimes \mathcal{K})$: 
\beq\label{eq:prod_met1}
d_\pi(U, V) = \max_{\rho \in S_{sep}(\mathcal{H}\otimes \mathcal{K})} d_\rho (U, V) = \max_{\rho \in S_{sep}(\mathcal{H}\otimes \mathcal{K})} d(U\rho \hconj{U}, V\rho \hconj{V})
\eeq
That is, we apply the unitary operators to {\em separable} states to compute the distance between them. The reason is, we want to study the {\em entangling capacity} of a unitary operator--- the amount of entanglement it can produce starting from a ``classical'' state of the composite system (created by LOCC). Now we can define the entangling capacity (or simply capacity) $C(U)$ of an operator  $U \in \mathcal{U}(\mathcal{H}\otimes \mathcal{K})$. For convenience of notation the underlying spaces $\mathcal{H}, \mathcal{K}$ will be dropped when they are clear from the context. Thus write $\Pi_u$ instead of $\Pi_u(\mathcal{H}\otimes \mathcal{K})$ and $S_{sep}$ instead of $S_{sep}(\mathcal{H}\otimes \mathcal{K})$. 

\beq \label{eq:entcap1} 
\begin{split}
C(U) & = \inf_{V\in \Pi_u} d_\pi(U, V) \\
& = \inf_{V\in \Pi_u} \sup_{\rho \in S_{sep}} \tt{Tr}(|U\rho \hconj{U} - V\rho \hconj{V})|)/2\\
& =  \min_{V\in \Pi_u} \max_{\rho \in S_{sep}} \tt{Tr}(|U\rho \hconj{U} - V\rho \hconj{V})|)/2\\
\end{split}
\eeq
We could replace {\em inf} and {\em sup} by {\em min} and {\em max} respectively, since we are dealing with compact sets. Moreover, the maximum will occur at some \emph{pure product} state (see the discussion after \eqref{def:psi-dist}). Hence, for the trace ($L_1$) norm 
\beq\label{eq:entcap2}
C(U) =  \min_{V\in \Pi_u} \max_{\psi \in \Pi_p} 
\sqrt{1 - \modulus{\inpr{\psi}{V^\dagger U}{\psi}}^2} 
\eeq
The following theorem gives some of the basic properties of entangling capacity. 
\begin{thm}
Let $U, V$ be unitary operators on $\mathcal{H}\otimes \mathcal{K}$. Then the following statements are true. 
\be[{\em a.}]
\item
$C(U)$ is invariant under multiplication (left {\em and right}) by unitary operators  $X\in \Pi_u$ (local operations). That is, 
\beq 
C(U) = C(XU) = C(UX)
\eeq
\item
$C(U) = C(\hconj{U}) $
\item
$C(UV) \leq C(U) + C(V)$. 
\item
$C(U)$ is a continuous function on the metric space $(\mathcal{U}, d_\pi)$. 
\ee
\end{thm}
\begin{proof}
\be[{\em a.}]
\item
The invariance of $C(U)$ under multiplication by local unitary operator $X$ follows from the fact that $d_\pi(U, V) = d_\pi(XU, XV) = d_\pi (UX, VX)$. The second equality is a consequence of the fact that $X = X_1\otimes X_2$ is a local unitary operation. So  
\[
\begin{split}
 & d_\pi (UX, VX)  \\
& = \max_{{\alpha\in \mathcal{H}},{\beta\in \mathcal{K}}} \left(1 -{\left|\bra{\alpha}\bra{\beta}\hconj{X_1}\otimes\hconj{X_2}\hconj{U}VX_1\otimes X_2
\ket{\alpha}\ket{\beta}\right|}^2\right)^{1/2}\\ 
& = d_\pi(U, V) \\
\end{split}
\] 
\item
Let $X \in \Pi_u$ be a  local unitary such that $C(U) = d_\pi(X, U)$. Then, 
\[ 
\begin{split}
C(\hconj{U}) &  \leq d_\pi(\hconj{X}, \hconj{U}) = d_\pi(\hconj{X}X, \hconj{U}X) \\ 
& = d_\pi (I, \hconj{U}X) = d_\pi(U, X) = C(U)\\
\end{split}
\]
Interchanging $U$ and $\hconj{U}$ above one gets $C(U) = C(\hconj{U})$. 
\item
Let $X, Y \in \Pi_u$ such that $C(U) = d_\pi(X, U)$ and $C(V) = d_\pi(Y, V)$. Then
\[
\begin{split}
C(UV) & \leq d_\pi(XY, UV) \leq d_\pi(XY, UY) + d_\pi(UV, UY) \\
& = d_\pi(X, U) + d_\pi(V, Y) = C(U) + C(V) \\
\end{split}
\]
The second step follows from the fact that $d_\pi$ is invariant under left multiplication by {\em any} unitary operator on $\mathcal{H}\otimes \mathcal{K}$. 
\item
Since the minimum value $C(U)$ is attained, we have $C(U) = d_\pi(V_1\otimes V_2, U)$ for some $V_1\otimes V_2 \in \Pi_u$. Then for any unitary $U' \in \mathcal{U}(\mathcal{H}\otimes \mathcal{K})$
\[
\begin{split}
C(U') & \leq d_\pi(V_1\otimes V_2, U') \\
& \leq d_\pi(V_1\otimes V_2, U) + d_\pi(U,U') = C(U) + d_\pi(U, U') 
\end{split}
\]
Interchanging the role of $U$ and $U'$ we get $|C(U) - C(U')| \leq d_\pi(U, U')$ proving continuity. 
\ee
\end{proof}
\noindent
Computing the value of $C(U)$ from the formulas \eqref{eq:entcap1} or \eqref{eq:entcap2} is rather difficult. So we try to estimate it indirectly. 
\subsection{General class of capacities} 
In Subsection \ref{subsec:gen1} the definition of metrics was modified to accommodate a general class of metrics. The change was effected by restricting the input set over which we maximize. We can extend the same idea to the \emph{output} side. In the definition \eqref{eq:entcap1} the minimization is to be done over $\Pi_u$, the group of product unitary operators and maximization over $S_{sep}$, the set of seprable states. We could generalize both. A restriction on the states would amount to a restriction on the input resources. Similarly a different set of operator from which we measure the distance is reflection on the output resources. So, let $\mathcal{O}$ be a closed subset $CP(\mathcal{H}\otimes \mathcal{K})$ (completely positive operators) and $\mathcal{G}$ a closed subset of $S(\mathcal{H}\otimes\mathcal{K})$. Define 
\beq \label{eq:entcapG} 
\begin{split}
C_{\mathcal{O},\mathcal{G}}(U) & = \inf_{L\in \mathcal{O}} d_\mathcal{G}(U, L) \\
& =  \min_{L\in \mathcal{O}} \max_{\rho \in \mathcal{G}} \tt{Tr}(|U\rho \hconj{U} - L\rho \hconj{L})|)/2\\
\end{split}
\eeq
One could take the set $\mathcal{O}$ to be closure of all LOCC operations. Alternatively, we could restrict $\mathcal{O}$ to be some (closed) subgroup of $\Pi_u$. For example, the subgroup of product unitary operators acting on a subspace $\mathcal{H}'\otimes\mathcal{K}'$, where $\mathcal{H}'$ (resp. $\mathcal{K}'$) is a subspace of $\mathcal{H}$ ($\mathcal{K}$). This would be appropriate if ancilla are used in the entanglement preparation stage. However, in this paper only the noiseless case without use of ancilla is studied. 
\subsection{The dual problem}
The formula \eqref{eq:entcap1} (or  \eqref{eq:entcap2}) for the entangling capacity of a unitary operator on a product space is an instance of {\em minimax} problem (\cite{Rockafellar}, \cite{Du}). The term ``minimax'' encodes the instruction: maximise and \emph{then} minimise.  A minimax problem has an associated dual problem--- a \emph{maximin} problem in which the 2 operations are reversed. We arrived at the expressions for $C(U)$ from the definition of entangling capacity. It is entirely reasonable to get  at an alternative definition. Thus, let $U$ be a unitary operator on $\mathcal{H}\otimes \mathcal{K}$ as before. Now  consider the image $\mathcal{Z} = U(\Pi_p(\mathcal{H}\otimes \mathcal{K}))$, that is, the set of pure states generated when $U$ is applied to the set of pure \emph{product} states. Some states in  $\mathcal{Z}$ will be entangled unless $U\in \Pi_u( \mathcal{H}\otimes  \mathcal{K})$. Suppose $\rho \rightarrow E(\rho)$ is an entanglement measure on states. We ask the question: what is the maximum $E(\rho), \;\rho \in \mathcal{Z}$. In other words what is the maximum entanglement, measured by $E$, of a state that can be generated by $U$ applied to pure product states? This leads to  an alternative  entanglement capacity 
\beq\label{eq:entcap_dual1}
C_E(U) = \max_{\alpha \in U(\Pi_p(\mathcal{H}\otimes \mathcal{K}))}E(\alpha) 
\eeq
Since  a geometric approach is followed in this paper, I choose $E$ to be a {\em geometric} entanglement measure \cite{Vedral, Chen}. In particular, let 
\[
E(\alpha) = \min_{\rho \in S_{sep}} \norm{\rho - \pj{\alpha}}_1/2 =  \min_{\rho \in S_{sep}} \tt{Tr}(|\rho - \pj{\alpha}|)/2
\]
$E$ is a (weak) entanglement measure on states \cite{Chen}. We also restrict $\rho$ to pure product states.\footnote{I cannot completely justify this. Can we assert that there is a pure state in $S_{sep}$ closest to a pure (entangled) state?} Then the equation \eqref{eq:entcap_dual1} maybe written as 
\beq\label{eq:entcap_dual2}
\begin{split}
C_E(U) &= \max_{\ket{\psi}\in \Pi_p} \min_{\ket{\phi} \in \Pi_p} {\tt Tr}(|\pj{\phi}- U\pj{\psi}\hconj{U})/2 \\
& = \max_{\ket{\psi}\in \Pi_p} \min_{\ket{\phi}\in \Pi_p} 
\sqrt{1 - | \bra{\phi}U\ket{\psi}|^2} \\
& = \max_{\ket{\psi}\in \Pi_p} \min_{V\in \Pi_u}  
\sqrt{1 - | \bra{\psi}\hconj{V}U\ket{\psi}|^2} \\
\end{split}
\eeq
The last step is justified by the transitive action of the group $\Pi_u = \mathcal{U}(\mathcal{H})\otimes \mathcal{U}(\mathcal{K})$ on product states. This means that for any pair pure product states, $\phi\text { and } \psi$, there is a product unitary operator $V$ such that $\phi = U\psi$. Note that in the formulas  \eqref{eq:entcap2} and \eqref{eq:entcap_dual2} the order of minimization and maximization is changed---the problems are \emph{dual} to each other. The first is a \emph{minimax} (maximization followed by minimization) and the  second a \emph{maximin} problem \cite{Rockafellar}
For a general formulation, Let $A\subset X \text{ and }B\subset Y$ be compact subsets  of finite-dimensional vector spaces $X \text{ and } Y$. Let $F:X\times Y \rightarrow \Real$ be a continuous function. Then a minimax problem requires to find the number 
\beq\label{eq:minimax}
M = \min_{x\in A}\max_{y\in B} F(x, y) 
\eeq
The dual of this problem ({\em maximin}) is to compute 
\beq\label{eq:maximin}
m = \max_{y\in B}\min_{x\in A} F(x, y) 
\eeq
A simple but important relation is the {\em minimax inequality}: $m\leq M$ \cite{Rockafellar, Du}. A major part of minimax theory is devoted to finding conditions under which the inequality becomes an equality. I will discuss this important issue later. In any case, it follows that 
\beq\label{eq:minimax1}
C_E(U) \leq C(U) \text{ for all } U\in \mathcal{U}(\mathcal{H}\otimes\mathcal{K})
\eeq

\subsection{Computing entangling capacities} 
Computing $C(U)$ directly from the definitions appears to be a difficult problem. So we approach it indirectly via the dual $C_E(U)$. Writing $\ket{\alpha}\ket{\beta}$ for $\alpha\otimes\beta$ define 
\beq\label{eq:fid}
\begin{split}
& G(\alpha, \beta, U_1, U_2; U) = |\bra{\alpha}\bra{\beta}\hconj{U_1}\otimes \hconj{U_2}U\ket{\alpha}\ket{\beta}|,\\
& \alpha\in \mathcal{H}, \beta \in \mathcal{K}, U_1\in \mathcal{U}(\mathcal{H})\text{ and } U_2\in \mathcal{U}(\mathcal{K})\\
\end{split}
\eeq
Then 
\begin{subequations}\label{eq:entcap3} 
\begin{align}
C(U) & = \sqrt{1 - (\max_{U_1, U_2}\min_{\alpha, \beta}G(\alpha, \beta, U_1, U_2; U))^2}\\
C_E(U) & = \sqrt{1 - (\min_{\alpha, \beta}\max_{U_1, U_2}G(\alpha, \beta, U_1, U_2; U))^2}
\end{align}
\end{subequations}
The vector $U(\alpha\otimes\beta)$ is in general an entangled state and one can find bases $\{\alpha_i\}, \{\beta_j\}$  in $\mathcal{H}\text{ and }\mathcal{K}$ respectively such that $ U(\alpha\otimes\beta) = \sum_{i=1}^r s_i \alpha_i\otimes \beta_i$, the number of nonzero terms being less than equal to smaller of the dimension (Schmidt decomposition). Further, the bases can be defined so that $s_i\geq 0$ \cite{Nielsen}. The following lemma exploits the Schmidt decomposition for computing $C_E(U)$. 
\begin{lem}
For a given product vector $\alpha\otimes\beta\in \mathcal{H}\otimes\mathcal{K}$ let 
\[ U(\alpha\otimes\beta) = \sum_{i=1}^{m} s_i \alpha_i\otimes \beta_i,\; s_i \geq 0, \sum_i s_i^2 = 1
\]
be the Schmidt decomposition. Let $\mathrm{dim}(\mathcal{H}) = m$ and $\mathrm{dim}(\mathcal{K}) = n$ and assume (without loss of generality)  $m\leq n$. Also suppose the coefficients are arranged in decreasing order $s_1 \geq s_2 \geq \dotsb s_m \geq 0$ then 
\[
\max_{U_1, U_2}G(\alpha, \beta, U_1, U_2; U) \leq s_1^2
\]
\end{lem}
\begin{proof}
Let $U_1\alpha= \sum_i a_i\alpha_i$ and $U_2\beta= \sum_i b_i\beta_i$. Then 
\begin{align*}
&G(\alpha, \beta, U_1, U_2; U)^2  = |\sum\conj{a}_i\conj{b}_i s_i|^2 \\
& \leq (\sum |a_ib_i|s_i )^2\leq s_1^2 (\sum |a_ib_i|)^2 \\
& \leq s_1^2\sum |a_i|^2 \cdot \sum |b_i|^2 = s_1^2
\end{align*}
The last inequality follows from the fact that $\alpha$, $\beta$ are unit vectors and the preceding one is a consequence of Cauchy-Schwartz inequality. 
\end{proof}
\noindent
If we choose $U_1, U_2$ such that $U_1\alpha = \alpha_1$ and $U_2\beta = \beta_1$ then $G(\alpha, \beta, U_1, U_2; U) = s_1^2$. 
It follows that $c \equiv  \min_{\alpha, \beta}\max_{U_1, U_2}G(\alpha, \beta, U_1, U_2; U)$ is the minimum value of the largest Schmidt coefficient of  $U(\alpha\otimes\beta)$ as we vary $\alpha, \beta$. Denote this minimum value by $s_{\mu}(U)$. Since $\sum_{i=1}^{m} s_i^2 = 1$, $s_1^2 \geq \frac{1}{m}$. We get the following:  
\beq \label{eq:entCap0}
C_E(U) = \sqrt{1 - s^2_{\mu}(U)}\leq \sqrt{1 -\frac{1}{m}}
\eeq
Note that the upper bound for $C_E(U)$ is achieved by the {\em maximally-entangled} states. If the maximum (Scmidt) rank of the the states $U(\alpha\otimes\beta)$ is $r$, as one varies $\alpha$ and $\beta$ then the above inequality can be improved.  
\begin{equation}\label{ineq:ent_state2}
C_E(U) \leq \sqrt{1 -\frac{1}{r}}
\end{equation}
It is considerably more difficult to give good estimates of $C(U)$ because it requires estimating non-trivial lower bounds. But it follows from the minimax inequality \eqref{eq:minimax1} that if a unitary operator $U$ generates a maximally entangled state ($C_E(U) = \sqrt{1-1/m}$) then 
\beq\label{ineq:cap_lowerB1}
 \sqrt{1-\frac{1}{m}} \leq C(U)
\eeq
\section{Generalised control unitary operators}
Control operators generate entangled states when applied to appropriate product states. The most familiar 2-qubit control gates like the CNOT ($C_X$) and controlled-$Z$ ($C_Z$) are basic to building quantum circuits \cite{Nielsen}. The operators  $C_X$ and $C_Z$ are defined in the ``computational'' basis as 
\[
\begin{split}
C_X\ket{0}\ket{\alpha}  & = \ket{0}\ket{\alpha}\text{ and } C_X\ket{1}\ket{\alpha}  = \ket{1}X\ket{\alpha}\\
C_Z\ket{0}\ket{\alpha}  & = \ket{0}\ket{\alpha}\text{ and } C_Z\ket{1}\ket{\alpha}  = \ket{1}Z\ket{\alpha}\\
\end{split}
\]
The corresponding matrices in this basis are 
\[
\begin{aligned}
& C_X  = \begin{pmatrix} I_2 & 0 \\ 0 & X^{(2)} 
\end{pmatrix}
& C_Z  = \begin{pmatrix} I_2 & 0 \\ 0 & Z^{(2)} 
\end{pmatrix}\\
& X^{(2)}  = \begin{pmatrix} 0 & 1\\ 1 & 0 \end{pmatrix} 
& Z^{(2)}  = \begin{pmatrix} 1 & 0\\ 0 & -1 \end{pmatrix} 
\end{aligned}
\]
Here, somewhat fussy notation for the 2-dimensional $X\text{ and } Z$ operators are used because I will be discussing these operators in arbitrary dimension. In the following subsection I propose a generalization of control unitary operators \footnote{The presentation that follows is new, to the best of my knowledge.}.

\subsection{Definitions and some properties of GCO} 
Control operators like $C_Z, C_X$ in 2-qubit space have the characteristic that one qubit, the control bit, is unchanged and depending on its state a unitary operator is applied to the other. 
\begin{Def*}
Let $\mathcal{H}$ and $\mathcal{K}$ be Hilbert spaces of dimension $m$ and $n$ respectively ($m\leq n$)\footnote{No loss of generality here!}. Let $\mathscr{A} = \{\alpha_1, \dotsc, \alpha_m \}$  be an orthonormal basis in  $\mathcal{H}$.  Suppose  $\mathscr{U}_f = \{U_1, \dotsc, U_m\}$ is a family of   unitary operators (not necessarily distinct) on $\mathcal{K}$.
\commentout{
\beq\label{eq:unitary_family}
\eeq
}
The {\em rank } of a unitary family $\mathscr{U}_f$ is defined as 
\beq\label{eq:uf_rank}
\mathrm{rank}(\mathscr{U}_f) = \max_{\beta\in \mathcal{K}}\{ k | \; k = \text{\em rank}(U_1\beta\; U_2\beta\; \dotsc  U_m\beta)\}
\eeq
In this equation $(U_1\beta\; U_2\beta\; \dotsc  U_m\beta)$ denotes the matrix with columns $U_i\beta$. The family $\mathscr{U}_f$ is called abelian if for all operators in $\mathscr{U}_f$ mutually commute. 
Define a  operator $U$ on $\mathcal{H}\otimes\mathcal{K}$ associated with a unitary family $U_f$ as follows. 
\beq\label{eq:gen_control1}
U(\alpha_i\otimes \beta) = \alpha_i\otimes U_i\beta
\eeq
Then $U$ is called a generalized control operator ({\bf GCO}) over unitary family $\mathscr{U}_f$ and  basis $\mathscr{A}$ . 
\end{Def*}

\noindent
The rank of a unitary family $\mathscr{U}_f$ is the maximal dimension of a subspace generated by $\{U_i\beta| \; U_i\in \mathscr{U}_f, \beta\in \mathcal{K}\}$. The operator $U$ defined by  \eqref{eq:gen_control1} is essentially unique except perhaps a permutation of $\mathscr{A}$ (or $\mathscr{U}_f$) . Some properties of GCO are proved in the following theorem. 
\begin{thm}\label{thm:gco1}
Given Hilbert spaces $\mathcal{H}$ and $\mathcal{K}$ of dimension $m$ and $n$ respectively ($m\leq n$) 
let $U$ be the operator on $\mathcal{H}\otimes\mathcal{K}$ defined over a unitary family $\mathscr{U}_f = \{U_1, \dotsc, U_m\}$ and  a basis $\mathscr{A}= \{\alpha_1, \dotsc, \alpha_m\}$ of $\mathcal{H}$.  
Let $r = \mathrm{rank}\mathscr{U}_f$ . Then the following assertions hold. 
\be
\item
The operator $U$  is unitary. 
\item
If $r > 1$ then $C(U) \geq C_E(U) > 0$. In particular, the operator $U$ generates an entangled with Schmidt number $r$. 
 If $\psi  = U(\alpha\otimes \beta )$
is a maximally entangled state generated by $U$ with Schmidt decomposition 
\[ \psi = \sum_{i=1}^{r}  s_i \alpha_i'\otimes\gamma_i ,\; s_1\geq s_2\geq \dotsb s_{r} >0 \text{ and } \sum s_i^2 = 1\]
then $C_E(U) = \sqrt{1- s_1^2}\leq \sqrt{1-1/r}\leq  \sqrt{1-1/m}$. 
\item
Assume that the family $\mathscr{U}_f$ is abelian and let the basis $\mathscr{B}$ consist of common eigenvector $U_i$ and let $U_i\beta_j = e^{i\pi \theta_{ij}} ,\; 0\leq i < m \text{ and } 0 \leq j < n$. Define the matrix $\Theta$ by $(\Theta)_{ij} = e^{i\pi \theta_{ij}}$. Then $r = \mathrm{rank}(\mathscr{U}_f) = \mathrm{rank}(\Theta)$. 

For unit vectors $\vec{a} = (a_1,  \dotsc, a_m) \in \Complex^m \text{ and }  \vec{b} = (b_1, \dotsc, b_n) \in \Complex^n$, let  
$A(\vec{a}, \vec{b})$ be the $m\times n$ matrix with entries is $a_{i}e^{i\pi \theta_{ij}}b_{j}$. 

The capacity $C_E$ can be formulated as the solution to the following minimax problem. 
\beq \label{eq_entCap2}
C_E(U) = \sqrt{1 - \min_{\vec{a}, \vec{b},\norm{\vec{a}}=\norm{\vec{b}} =1} \max_{\vec{x}, \norm{\vec{x}} =1} \norm{Ax}^2}
\eeq
\item
The operator $U$ generates a  maximally entangled state $(C_E(U) =  \sqrt{1 - 1/m})$ in $\mathcal{H}\otimes\mathcal{K}$  if and only 
if there is a $\beta \in \mathcal{K}$ such that the vectors $\{U_1\beta,  \dotsc, U_m\beta\}$ are mutually orthogonal. Then the maximally entangled state  is of the form 
\[ U\left(\frac{\alpha_1+\dotsb + \alpha_{m}}{\sqrt{m}}\otimes \beta\right) = \frac{1}{\sqrt m} \sum_{i=1}^{m} \alpha_i\otimes U_i\beta\]
\ee
\end{thm}

\begin{proof}
Let $\mathscr{B}= \{\beta_1, \dotsc, \beta_{n}\}$ be an orthonormal basis in $\mathcal{K}$. To prove that $U$ is unitary it suffices to show that $\{U(\alpha_i\otimes \beta_j)  = \alpha_i \otimes U_i\beta_j|\; 1\leq i \leq m\text{ and } 1\leq j \leq n\}$ is also orthonormal. For $i\neq k$, $\alpha_i \otimes U_i\beta_j$ and $\alpha_k \otimes U_k\beta_j$ are obviously orthogonal. $\alpha_i \otimes U_i\beta_j$ and $\alpha_i \otimes U_i\beta_l$ are orthogonal because the $U_i$ are unitary operators. 

For the second item, suppose $\beta\in \mathcal{K}$ is a vector such that the set $\{U_i\beta, \; 1\leq i \leq m\}$ has maximal rank $r$. Assume, without loss of generality, that $\{U_1\beta, \dotsc, U_r\beta\}$ is an independent set and let 
\[ U_i\beta = \sum_{j=1}^{n} b_{ij} \beta_j, \; 1\leq i \leq m \text{ and } \alpha = \frac{\sum_{i=1}^{r}\alpha_i }{\sqrt{r}}\] 
Then $\xi = U(\alpha\otimes \beta) = \sum_{ij} b_{ij}\alpha_i\otimes \beta_j,\; 1\leq i \leq r, \; 1\leq j \leq n$ has $r\times n$ coefficient matrix $B$ of full rank $(=r)$. So $B$ has $r$ {\em positive} singular values $t_1 \geq t_2 \dotsb \geq t_r >0$ \cite{HJ1, Bhatia} and there exist bases $\{\alpha'_1, \dotsc, \alpha'_m\}$ $(\{\beta'_1, \dotsc, \beta'_{n}\})$ in $\mathcal{H}$ $(\mathcal{K})$ such that 
\[ U(\alpha\otimes\beta) = \sum_{i=1}^{r} t_i \alpha'_i\otimes \beta'_i \]
The rest of the assertions in item 2 follow from the relations \eqref{eq:entCap0}. 

Next, suppose the family $\mathscr{U}_f$ is abelian. Then there is a basis $\mathscr{B}$ in $\mathcal{K}$ consisting of \emph{common} eigenvectors of operators in  $\mathscr{U}_f$.  If $\alpha 
=\sum_ia_i\alpha_i$ and $\beta = \sum_jb_j\beta_j$ with $U_i\beta_j = e^{i\pi \theta_{ij}} = (\Theta)_{ij}$. Then 
\[ U(\alpha\otimes \beta) = \sum_{ij} a_i b_j\theta_{ij} \alpha_i\otimes\beta_i \] 
Let $D_{\vec{a}}$ be the diagonal matrix with entries $a_i$ similarly $D_{\vec{b}}$ be the diagonal matrix with entries $b_i$. We can write 
\beq \label{eq:coeff}
A(\vec{a}, \vec{b}) =  D_{\vec{a}}\Theta  D_{\vec{b}}
\eeq
where $\vec{a}$ (resp. $\vec{b}$) is the vector whose coordinates are $a_i$ (resp. $b_i$). 
It is easy to see that $r = \max_{\vec{a}, \vec{b}} \mathrm{rank}(D_{\vec{a}}\Theta  D_{\vec{b}})$. But for any $\vec{a}, \vec{b}$, $ \mathrm{rank}(D_{\vec{a}}\Theta D_{\vec{b}}) \leq  \mathrm{rank}(\Theta)$. The equality is attained if   $D_{\vec{b}} \text{ and } D_{\vec{b}}$ are nonsingular which is easily arranged. The first part of the assertion is proved. We also see that the rank of the coefficient matrix $A$ is $\leq r$. The second half of item 3 follows from the fact that the largest singular value $s_1$ of a matrix $M$ is given by \cite{HJ1, Bhatia} $s_1 = \max_{x,\norm{x}=1} \norm{Mx} $

For the final item suppose $\text{rank}(\mathscr{U}_f) = m$ and $\{U_1\beta,  \dotsc, U_m\beta\}$ are mutually orthogonal for some unit vector $\beta \in \mathcal{K}$. Let $\alpha = (\alpha_1+\dotsb + \alpha_m)/\sqrt{m}$. Then the vector $U(\alpha\otimes \beta) = (\alpha_1\otimes U_1\beta+\dotsb + \alpha_m\otimes U_m\beta)/\sqrt{m}$ is already in the  Schmidt form and $C_E(U) = \sqrt{1-1/m}$. 

To prove the converse, assume $U$ generates a maximally entangled state $\psi = U(\alpha\otimes \beta)$. 
From item 2 it follows that   the set $S =\{ U_1\beta, \dotsc, U_m\beta\} $ is linearly independent. Using Gram-Schmidt orthogonalization \cite{HJ1} obtain an orthonormal basis $\beta'_1, \dotsc, \beta'_m$ in the subspace spanned by $S$. So $U_i\beta = \sum_{j=1}^m b_{ij} \beta'_j, \; i,=1, \dotsc, m$. If $\alpha = a_1\alpha_1 + a_2\alpha_2 + \dotsb+a_m\alpha_m$ (note, $a_i \neq 0$ for all $i$). 
\[ \psi = U(\alpha\otimes \beta) = \sum_{i=1}^m a_i \alpha_i\otimes U_i\beta = \sum_{i,j=1}^m a_ib_{ij} \alpha_i\otimes\beta'_{j} \]
Note that the coefficient matrix $A$ with entries $a_ib_{ij}$ is square of order $m$. Since $\psi$ is a maximally entangled state its Schmidt decomposition is of the form
\[ \psi = \frac{\alpha'_1\otimes\gamma_1 + \alpha'_2 \otimes \gamma_2 + \dotsb + \alpha'_m\otimes \gamma_m}{\sqrt m} \]
The Schmidt coefficients are singular values of the matrix $A$ \cite{Nielsen}. Hence there exist unitary matrices $V, W$ of order $m$ such that $A = VDW$ where $D$ is a diagonal matrix consisting of singular values of $A$. Since the singular values all equal $1/\sqrt{m}$, $D = I_m /\sqrt m$ where $I_m$ is the identity matrix of order $m$. Hence $A= VW/\sqrt m$. 
Write 
\[ D_{\vec{a}} = \begin{pmatrix}
a_1 &0 \cdots & 0 \\
0 &\ddots & 0 &\\
0 & \cdots &  a_m &
 \end{pmatrix}, \; 
 B =\begin{pmatrix}
b_{11} & \cdots & b_{1m} \\
\vdots &\vdots &\vdots\\
b_{m1} & \cdots &  b_{mm} 
 \end{pmatrix}
 \]
Then 
\[ A= D_{\vec{a}}B = \frac{VW}{\sqrt m} \text{ and }B = \frac{ D_{\vec{a}}^{-1}}{\sqrt m} \cdot VW\]
Since $D_{\vec{a}}$ is diagonal every row of $B$ is a multiple of the corresponding row of $VW$. The latter is a unitary matrix. So the rows of $B$ are mutually orthogonal. Recall that the $i^{\text{th}}$ row of $B$ represents $U_i\beta$ in the {\em orthonormal} basis $\{\beta'_1, \dotsc, \beta'_m\}$ and hence $\sum_j \modulus{b_{ij}}^2 = 1$. Therefore, $a_i = 1/\sqrt m$ and 
\[ \inp{U_i\beta}{U_j\beta} = \delta_{ij} \]
The proof is complete. 
\end{proof} 
\noindent
A different approach to generating maximally entangled states using representations of finite groups is given in reference \cite{Cohen}. 
\subsection{Examples} 
This subsection provides some examples of GCO. First observe that for an arbitrary unitary family $\{U_1, \dotsc, U_m\}$ if we replace all operators with $U_i \longrightarrow \hconj{U}_1 U_i$, the corresponding GCO, $U$   does not change essentially (e.g. entangling capacities remain same). So  take $U_1 = I$  (identity operator in appropriate space). 
The computations are manageable when $m=n=2$, so we start there. It is well-known that the 2-qubit case is special. For example, for 2 qubits the PPT (positive partial state) criterion is necessary {\em and} sufficient to determine separability, not true in higher dimensions. 
\begin{enumerate}
\item
Suppose $\{\alpha_1, \alpha_2\}$ is an  arbitrary orthonormal basis $\mathcal{H}_2$. The general form of a GCO is given by
\[ U(\alpha_1\otimes \beta) = \alpha_1\otimes\beta \text{ and } U(\alpha_2\otimes \beta) = \alpha_2\otimes U_2\beta,\; \beta \in \mathcal{H}_2 \]
The projective invariance of the underlying metric implies that we may take $U_2$ to be of the form 
\[ U_2 = \begin{pmatrix} 1 & 0 \\ 0 & e^{i\pi \theta} \end{pmatrix},\; 0 \leq \theta \leq 1 \]
in the standard basis $\epsilon_1, \epsilon_2$. Assume that $\alpha_1$ and $\alpha_2$ also represent the standard basis (no loss of generality). 
So for \(\alpha = \begin{pmatrix} a_0\\a_1 \end{pmatrix} \; \beta  = \begin{pmatrix} b_0\\b_1 \end{pmatrix}\),
\[U(\alpha\otimes \beta) = a_1b_1 \epsilon_1\otimes\epsilon_1 + a_1b_2 \epsilon_1\otimes\epsilon_2 
+a_2b_1  \epsilon_2\otimes\epsilon_1 + a_2b_2  e^{i\pi \theta}  \epsilon_2\otimes\epsilon_2
\]
A little analysis shows that one may assume $a_1, a_2, b_1, b_2 \geq 0$. The coefficient matrix is 
\[ A = \begin{pmatrix} a_1b_1 & a_1b_2 \\a_2b_1 & a_2b_2  e^{i\pi \theta} \end{pmatrix}, \; a_1^2 + a_2^2 = b_1^2+b_2^2 = 1 \]
A tedious calculation shows that the  singular values of $A$ are 
\[s_1, s_2 = \sqrt{\frac{1 \pm \sqrt{1  -16a_1^2a_2^2b_1^2b_2^2\sin^2(\pi\theta/2)}}{2}}\]
Therefore, the larger singular value $s_1$ is minimal when $a_1^2=b_1^2=1/2$. This minimal value is $\frac{1+\cos{(\pi\theta/2)}}{2}$. So, $C_E(U) = \sqrt{\frac{1-\cos(\pi\theta/2)}{2}}$ attains its maximum value $\sqrt{1/2}$ when $\theta=1$. This is as expected, for then $U_2$ acts as Pauli operator $Z$ in some basis in 2 dimensions. 
\item
Let $\mathcal{H} = \mathcal{K}$. Assume that $\mathscr{A} = \{\alpha_1, \dotsc, \alpha_n\}$ is an orthonormal basis in  $\mathcal{H}$ and $Z$ be quantum Fourier transform operator (QFT), $Z\alpha_j = \omega^{j-1}\alpha_j, \; j=1, \dotsc, n$,  where $\omega = e^{\frac{2\pi i}{n}}$ is a primitive $n^{\text{th}}$ root of 1. Define an operator $U$ over  {\em abelian} unitary family $ \mathscr{U}_f = \{ I=Z^0, Z, \dotsc, Z^{n-1}\} $ by 
\[ U(\alpha_j\otimes \beta) = \alpha_j \otimes Z^{j-1}\beta, \; i=1, \dotsc, n \text{ and } \beta \in \mathcal{H} \]
Let $\beta = \frac{\sum_i \alpha_i}{\sqrt n}$. Since $\inp{Z^i\beta}{Z^j\beta} =\delta_{ij}$, $U(\beta\otimes \beta)$ is maximally entangled (item 4 in Theorem \ref{thm:gco1}). 
\item
This example is a nonabelian generalization of the previous one. Let $m=n=3$. and $\{\alpha_1, \alpha_2, \alpha_3\}$ be a basis in $\mathcal{H}= \mathcal{K}$. The operator $Z$ is as defined in the previous example with $\omega = e^{\frac{2\pi i}{3}}$ and $X$ is the cyclic shift operator 
\[X\alpha_1=\alpha_2, \; X\alpha_2 =\alpha_3 \text{ and } X\alpha_3 = \alpha_1 \] 
Define a GCO $U$, 
\[ U(\alpha_1\otimes \beta) = \alpha_1\otimes \beta,\;  U(\alpha_2\otimes \beta) = \alpha_2\otimes Z\beta \text{ and }
U(\alpha_3\otimes \beta) = \alpha_3\otimes X\beta \]
The unitary operator $U$, although defined overs a nonabeliam family, has maximal capacity $C_E(U) = \sqrt{2/3}$. That is, it generates a maximally entangled state. Let $\beta = (\alpha_1 + \omega^2\alpha_2 + \omega^2\alpha_3)/\sqrt{3}$. Then $\beta_1 =\beta, \beta_2 = Z\beta \text{ and } \beta_3 = X\beta$ constitute an orthonormal triple. Therefore, state
\[ U(\frac{\alpha_1+\alpha_2 +\alpha_3}{\sqrt 3}\otimes \beta) = \frac{\alpha_1\otimes\beta_1 + \alpha_2\otimes\beta_2+\alpha_3\otimes\beta_3}{\sqrt3} \] 
is maximally entangled.  
\end{enumerate}
\section{Minimax duality and other questions} 
The entangling capacity $C(U)$ of a unitary operator $U$ on $\mathcal{H}\otimes \mathcal{K}$ is defined by  equation \eqref{eq:entcap1}. It depends on the underlying trace ($L_1$) norm. The simpler expression \eqref{eq:entcap2} is possible because the maximum occurs at an extreme point of the convex set of states. The dual capacity $C_E(U)$ is, by definition, the distance of the generated state to the manifold of {\em pure} product states. The resulting quantity in \eqref{eq:entcap_dual2} is dual to $C_U$. But if  we define 
\beq \label{eq:entcap_dual3}
\tilde{C}_E(U) = \max_{\ket{\psi}\in \Pi_p} \min_{\rho \in \mathcal{S}} \tt{Tr}(|\pj{\phi}- U\rho\hconj{U}|/2) 
\eeq
over separable states $\mathcal{S}$ there is no guarantee that it will reduce to the expression in \eqref{eq:entcap_dual2}. The basic reason is that the separable state closest to an entangled {\em pure} state may not be a pure (product) state\footnote{This {\em is} the case in case of metrics defined via {\em quantum fidelity}, for example, metrics induced by Bures distance.}. But we deal only with the measure $C_E(U)$ here. Recall the minimax inequality 
\[ C_E(U) \leq C(U) \]
A basic question in minimax problems is to seek conditions under which the above inequality becomes an equality. Many  sufficient conditions involve convexity/concavity properties (see \cite{Rockafellar} for some classical results and \cite{Du} for a variety of extensions). Unfortunately, some sets of interest to the present problem are not convex. However, it is possible to get some concrete results in low dimensions.  Let me recap the process of computing $C(U)$ for some unitary operator acting on a product Hilbert space. We have to find a product unitary  $V_1\otimes V_2$ such that the distance 
\[ d_\pi(U, V_1\otimes V_2) = \max_{\psi \in \Pi_p} \sqrt{1 - \modulus{\inpr{\psi}{U^\dagger V_1\otimes V_2}{\psi}}^2}\]
is minimized. For fixed $V_1\otimes V_2$, maximizing the right hand side seems difficult because of the restriction $\psi \in \Pi_p$, the set of pure product states. the result in theorem \ref{thm:eigenChar} cannot be used, for in general  $U$ may have eigenstates which are {\em not} product states. However, in case the operator $U$ does have a complete set of orthonormal eigenstates (a basis) then the theorem may be applied under some conditions. Let us start with control-Z gate in a 2-qubit setting. In appropriate basis (here assumed to be standard) 
\[
U_Z = \begin{pmatrix} I & 0 \\ 0 & Z \\  \end{pmatrix}, \; Z = \begin{pmatrix} 1 & 0 \\0 & -1 \end{pmatrix}
\]
Here $I$ is the identity matrix in 2 dimensions. Let 
\[ V_1 = V_2 = \begin{pmatrix} 1 & 0 \\ 0 & i \end{pmatrix} \text{ and } V_1 \otimes V_1 = \begin{pmatrix} 1 & 0 & 0 & 0 \\ 0 & i & 0 & 0 \\ 0 & 0 & i & 0 \\ 0 & 0 & 0 & -1 \end{pmatrix} \]
Making $V_1\otimes V_2$ diagonal in the (product) eigenbasis of $U_Z$ is a simple way of ensuring that $\hconj{U_Z}V_1\otimes V_2$ also has an orthonormal basis of product  eigenstates. In the present case, we can apply theorem \ref{thm:eigenChar} to obtain 
$ d_\pi(U_Z, V_1\otimes V_2) = \sin{\pi/4} = 1/\sqrt{2}= C_E(U_Z)$. But this implies that $C(U_Z) = C_E(U_Z) = 1/\sqrt{2}$. With a bit of extra calculations it can be shown that for any unitary $W$ in 2-d, the 2-qubit operator $U_W = \begin{pmatrix} I & 0 \\ 0 & W \end{pmatrix}$ also satisfies $C(U_W) = C_E(U_W)$. But of couse,  the 2-qubit case is special. In fact, if $\dim{(\mathcal{H})} = \dim{(\mathcal{K})} = 3$ and $U_Z = \begin{pmatrix} I & 0 & 0 \\ 0 & Z & 0 \\ 0 & 0 & Z^2 \end{pmatrix} $ ($Z$ is the QFT operator in 3 dimensions defined in the examples above) then it is possible to prove the following. For any product operator $V_1\otimes V_2$ \emph{commuting} with $U$, $d_\pi(U, V_1\otimes V_2) > C_E(U_Z) = 1/\sqrt{3}$. It is possible that $U_Z$ is better approximated by a product operator which does {\em not} commute with it. 

There are extensive results on duality in minimax theories. Many of the sufficient conditions for equality in the minimax inequality involve some notions of convexity (see \cite{Du}). Further investigations along these lines would require deeper analysis of the functions involved and their domains. Perhaps more fruitful from quantum physics and information perspective,  is   to extend the analysis to more general quantum operations, in particular CP maps. The general definition of entangling capacity given in \eqref{eq:entcapG} may be used to capture wide range of operational situations.  Questions of duality in minimax theories may be better addressed in the general settings because now we can choose sets with  some nice structures. 

Another possible avenue for exploration is the study of time evolution of an interacting composite system. For such systems the entangling capacities can be used as a measure of the strength of coupling. Now, given a coupled system with an interaction Hamiltonian $H$, we can obtain the corresponding evolution operator $U(t, t_0)$. For a closed system, this is a unitary operator and the capacities $C(U)\text{ and } C_E(U)$ become functions of time. So we may ask how do the capacities change with time? Or equivalently, how does the coupling strength change with time? Note that in this situation we may have to revisit the definition in \eqref{eq:entcap1}. Perhaps, more interesting would be to investigate evolution in an open system. Now the system is coupled to an environment. In a Hamiltonian setting the entangling capacity as function of time will shed light on the effect of environment over time. But here too the basic definitions should be re-examined.  Let us finally note that minmax problems are generally hard, but some recent numerical algorithms\cite{Lin} can be used for efficient numerical computations, at least in lower dimensions. 

\bibliography{entanglement}
\bibliographystyle{unsrt}
\end{document}